\documentclass[12pt,a4paper]{amsart}
\usepackage{amsfonts,amssymb,latexsym,amsmath, amsxtra, eufrak}
\usepackage[all]{xy}
\usepackage{graphics,graphicx,epsfig}

\usepackage[OT2,T1]{fontenc}
\DeclareSymbolFont{cyrletters}{OT2}{wncyr}{m}{n}
\DeclareMathSymbol{\Sha}{\mathalpha}{cyrletters}{"58}

\theoremstyle{plain}
\newtheorem{theorem}{Theorem}[section]

\newtheorem{lemma}[theorem]{Lemma}

\newtheorem{proposition}[theorem]{Proposition}

\newtheorem*{conjecture*}{Conjecture}

\newtheorem{definition}[theorem]{Definition}

\numberwithin{equation}{section}

\let\non\nonumber

\newcommand{\bea}{\begin{eqnarray}}
\newcommand{\eea}{\end{eqnarray}}
\newcommand{\be}{\begin{equation}}
\newcommand{\ee}{\end{equation}}

\newcommand{\sgn}{\mathrm{sgn}}

\newcommand{\noi}{\noindent}

\newcommand{\bP}{\mathbb{P}}

\newcommand{\half}{\textstyle{\frac{1}{2}}}

\newcommand{\CC}{\mathcal{C}}

\newcommand{\CM}{\mathcal{M}}
\newcommand{\CN}{\mathcal{N}}
\newcommand{\CO}{\mathcal{O}}


 
\begin{document}

\title[The Betti numbers of the moduli space of stable sheaves of rank 3 on $\bP^2$]{The Betti numbers of the moduli space of stable sheaves of rank 3 on $\bP^2$}

\author{Jan Manschot}
\address{Institut de Physique Th\'eorique\\
 CEA Saclay, CNRS-URA 2306\\
91191 Gif sur Yvette, France}
\email{jan.manschot@cea.fr}

\begin{abstract}
\baselineskip=18pt
\noi This article computes the generating functions of the Betti numbers of the moduli
space of stable sheaves of rank 3 on $\mathbb{P}^2$ and its blow-up
$\mathbb{\tilde P}^2$. Wall-crossing is used to obtain the Betti numbers for
$\mathbb{\tilde P}^2$. These can be derived equivalently using flow
trees, which appear in the physics of BPS-states. The Betti numbers for $\mathbb{P}^2$ follow from those for
$\mathbb{\tilde P}^2$ by the blow-up formula. 
The generating functions are expressed in terms of modular functions and indefinite theta
functions. 
\end{abstract}
\maketitle

\baselineskip=19pt


\section{Introduction}
The Euler and Betti numbers of moduli spaces of stable sheaves on
complex surfaces have received much attention in the past, both in
mathematics and physics. Computation of the generating functions
of these numbers is notoriously difficult; a generic result is only
known for rank 1 sheaves \cite{Gottsche:1990}. Yoshioka
\cite{Yoshioka:1994, Yoshioka:1995} has computed the generating 
functions of the Betti numbers for rank 2 sheaves on a ruled
surface using wall-crossing. The generating functions for rank 
2 sheaves on the projective plane $\mathbb{P}^2$ follow from these by
the blow-up formula \cite{Yoshioka:1994, Yoshioka:1995}. Some
qualitative differences appear for rank 3, in which case only three Poincar\'e
polynomials are known \cite{Yoshioka:1996}. Using a different
approach, namely toric geometry, Refs. \cite{Klyachko:1991,
  Kool:2009, weist:2009} have computed generating functions 
of the Euler numbers for sheaves of rank 2 and 3 on $\bP^2$.

The connection between sheaves and $\CN=4$ supersymmetric gauge theory relates the
Euler number of the moduli space to the supersymmetric index (or
BPS-invariant) \cite{Vafa:1994tf}. Another consequence of this connection is that the
$SL(2,\mathbb{Z})$ electric-magnetic duality group of gauge theory implies modular properties for the
generating functions of the Euler numbers \cite{Vafa:1994tf}, which was verified among
others for sheaves on $\bP^2$ with rank 1 and 2 \cite{Gottsche:1990, Yoshioka:1994}.  
Ref. \cite{Bringmann:2010sd} identified the modular properties for the
Betti numbers, which correspond to refined BPS-invariants
\cite{Dimofte:2009bv}. Although modularity has 
been proven useful for computations for rational surfaces \cite{Minahan:1998vr,
  Yoshioka:1998ti, Gottsche:1998}, mathematical justifications of
physical expectations, in particular a so-called holomorphic anomaly
\cite{Vafa:1994tf, Minahan:1998vr}, have been limited to rank $\leq$2 \cite{Gottsche:1990, Yoshioka:1994,
  Yoshioka:1998ti} since generating functions for $r>2$ were not known.    

This motivated the present work, which computes the generating functions of the 
Betti numbers for stable sheaves with rank 3 on the rationally ruled surface
$\mathbb{\tilde P}^2$ and on $\bP^2$. Using the results from \cite{Manschot:2009ia, Manschot:2010xp, Bringmann:2010sd},
the generating functions take a particularly compact form in terms of
modular functions and indefinite theta functions. The latter are
convergent sums over a subset of an indefinite lattice 
\cite{Gottsche:1996}. The lattice for rank 2 has
signature $(1,1)$ and the corresponding functions are well-studied in the literature
\cite{Gottsche:1996, Zwegers:2000}. Interestingly, the lattice for rank 3 has signature $(2,2)$
and the corresponding function is of a novel form. A detailed discussion
of the modular properties of this function will appear in a future article \cite{Bringmann:2011}.
 
The computations in this article rely on wall-crossing and the blow-up
formula, analogously to the computations by Yoshioka for rank
2, and can be extended to rank $>3$ if desired. To arrive at the generating functions, the (semi-primitive)
wall-crossing formula of \cite{Denef:2007vg, Dimofte:2009bv} is
applied to determine the change of the Betti numbers across a wall of marginal stability.
This wall-crossing formula is derived in physics both for
field theory \cite{Gaiotto:2008cd} and for supersymmetric black holes
in supergravity \cite{Denef:2007vg, Andriyash:2010qv}, which correspond to D$p$-branes (or sheaves) supported on $p$-cycles of a
Calabi-Yau threefold. The validity of the wall-crossing
formula for this article can be motivated by viewing the surface as a (rigid) divisor of a
Calabi-Yau threefold, and it is further confirmed by agreement of the
generating functions with older results in the mathematical literature \cite{Yoshioka:1996,
  Kool:2009, weist:2009}.  It is well-known that this wall-crossing
formula is equivalent with the mathematical wall-crossing formulas derived for Donaldson-Thomas
invariants \cite{Joyce:2008, Kontsevich:2008}. In fact, also
mathematical arguments exist that these formulas are
 applicable for invariants of moduli spaces of sheaves on $\tilde \bP^2$ since
$K^{-1}_{\tilde \bP^2}$ is numerically effective \cite{Joyce:2004, Kontsevich:2010px}. 

Wall-crossing for sheaves (or D4-branes) supported on divisors in
Calabi-Yau threefolds \cite{Manschot:2009ia} is another 
motivation for this paper. Considering the sheaves on surfaces without
the embedding into a Calabi-Yau simplifies the system and in this way
helps to understand the Calabi-Yau case. On the other hand supergravity gives a useful
complementary viewpoint, and suggests for example that 
the generating function for BPS-invariants can be computed  using enumeration of so-called flow trees
  \cite{Denef:2000nb}. This approach was taken in
  \cite{Manschot:2010xp}, and the present article provides an
  illustration and confirmation of this technique. Subsection
  \ref{subsec:flowtrees} gives a brief introduction to flow trees, however
  the discussion in this article is mostly phrased in terms of sheaves and characteristic classes,
  because the notion of a moduli space is most rigorously defined in
  this context. 
\\
\newline
The outline of the paper is as follows. Section \ref{sec:sheaves}
reviews the necessary properties of sheaves, including wall-crossing
and blow-up formulas. Subsection \ref{subsec:flowtrees} gives a brief
introduction to flow trees. Section \ref{sec:euler} computes the
Euler numbers of the moduli spaces for rank 2 and 3, followed by the computation of the
Betti numbers in Section \ref{sec:betti}. Appendix \ref{app:modfunctions}
lists various modular functions, which appear in the generating functions
of the Euler and Betti numbers. 

\section*{Acknowledgements}
\noi I would like to thank Emanuel Diaconescu, Lothar G\"ottsche and
Maxim Kontsevich for helpful and inspiring discussions, and IHES for hospitality.
This work is  partially supported  by ANR grant
BLAN06-3-137168.

\section{Sheaves}
\label{sec:sheaves}

\subsection{Sheaves and stability}
The Chern character of a sheaf $F$ on a surface $S$ is given by
ch$(F)=r(F)+c_1(F)+\frac{1}{2}c_1(F)^2-c_2(F)$ in terms of the
rank  $r(F)$ and its Chern classes $c_1(F)$ and $c_2(F)$. It is
convenient to parametrize a sheaf by ch$(F)$ since it is additive:
$\mathrm{ch}(F\oplus G) =\mathrm{ch}(F)+\mathrm{ch}(G)$. Define
$\Gamma:=(r,\mathrm{ch}_1,\mathrm{ch}_2)$. Other frequently occuring quantities are the determinant
$\Delta(F)=\frac{1}{r(F)}(c_2(F)-\frac{r(F)-1}{2r(F)}c_1(F)^2)$, and
$\mu(F)=c_1(F)/r(F)\in H^2(S,\mathbb{Q})$. 

Let $0\subset F_1 \subset F_2\subset \dots \subset F_s=F$ be a
filtration of the sheaf $F$. The quotients are denoted
by $E_i=F_i/F_{i-1}$ with $\Gamma_i=\Gamma(E_i)$. 

\begin{lemma}
With the above notation, the discriminant $\Delta(F)$ is given by
\be
\Delta(F)=\sum_{i=1}^s\frac{r(E_i)}{r(F)}\Delta(E_i)-\frac{1}{2r(F)}\sum_{i=2}^s
\frac{r(F_{i-1})\,r(F_i)}{r(E_i)} \left(\mu(F_{i-1})-\mu(F_i) \right)^2.\non
\ee
\end{lemma}

\begin{proof}
Consider first the filtration for $s=2$: $0\subset F_1\subset
F_2=F$, such that $\Gamma(F)=\Gamma(F_1)+\Gamma(E_2)$. Application of the definitions
and some straightforward algebra lead to:  
$$
\Delta(F)=\frac{r(F_1)}{r(F)}\Delta(F_1)+\frac{r(E_2)}{r(F)}\Delta(E_2)-\frac{r(F_1)r(E_2)}{2r(F)^2}(\mu(F_1)-\mu(E_2))^2.
$$
Applying this equation iteratively on $F_1$ leads to the
lemma.\footnote{\, Note
that this is different from Ref. \cite{Yoshioka:1996} (Lemma 2.2).}
\end{proof}

The notion of a moduli space for sheaves is only well defined after
the introduction of a stability condition. To this end let $C(S)\in
H^2(S,\mathbb{Z})$ be the ample cone of $S$.

\begin{definition}
Given a choice $J\in C(S)$, a sheaf $F$ is called $\mu$-stable if for every subsheaf $F'$,
$\mu(F')\cdot J <\mu(F)\cdot J$, and $\mu$-semi-stable if for every subsheaf $F'$,
$\mu(F')\cdot J \leq\mu(F)\cdot J$. A wall of marginal stability $W$ is a (codimension
1) subspace of $C(S)$, such that $(\mu(F')-\mu(F))\cdot J=0$, but
$(\mu(F')-\mu(F))\cdot J\neq 0$ away from $W$. 
\end{definition}

Let $S$ be a K\"ahler surface, whose intersection pairing on
$H^2(S,\mathbb{Z})$ has signature $(1,b_2-1)$.  Since at a wall,
$(\mu_2-\mu_1)\cdot J=0$ for $J$ ample,
$(\mu_2-\mu_1)^2<0$. Therefore, the set of filtrations for $F$, with
$\Delta_i\geq 0$ is finite.  

\subsection{Invariants and wall-crossing}
Ref. \cite{Vafa:1994tf} shows that the BPS-invariant of $\CN=4$ gauge
theory on $S$ equals the Euler number (up to a sign) of a suitable compactification of the
instanton moduli space, i.e. the Gieseker-Maruyama moduli space
$\CM_J(\Gamma)$ of semi-stable sheaves on $S$ (with respect
to the ample class $J$). The topological classes $\Gamma$ of
the sheaf are determined by the topological properties of the instanton.
The complex dimension of $\CM_J(\Gamma)$ is given by: 
\be
\dim_{\mathbb{C}} \CM_J(\Gamma)=2r^2\Delta-r^2\chi(\CO_S)+1.\non
\ee
The BPS-invariant $\Omega(\Gamma;J)$ corresponds to the Euler number of $\CM_{J}(\Gamma)$  
\be
\Omega(\Gamma;J)=(-1)^{\dim_\mathbb{C}\CM_J(\Gamma)}\chi(\CM_J(\Gamma)),\non
\ee
if the moduli space is smooth and compact. The mathematical rigorous
definition of the BPS-invariant is more involved if these
conditions are not satisfied \cite{Yoshioka:1995}. 
The rational invariants \cite{Nakajima:2007, Joyce:2008, Kontsevich:2008}
\be
\label{eq:ratinv}
\bar \Omega(\Gamma;J)=\sum_{m|\Gamma} \frac{\Omega(\Gamma/m;J)}{m^2}
\ee
are also particularly useful for our purposes
\cite{Manschot:2010xp}.

To state the changes $\Delta\Omega(\Gamma;J_\CC\to J_{\CC'})$ across walls of marginal
stability, we define the following quantities: 
\be
\label{eq:intprod} 
\left<\Gamma_1,\Gamma_2\right>=r_1r_2(\mu_2-\mu_1)\cdot K_S , \qquad \mathcal{I}(\Gamma_1,\Gamma_2;J)=r_1r_2(\mu_2-\mu_1)\cdot J.
\ee
These definitions follow quite naturally from formulas in physics \cite{Diaconescu:2007bf, Manschot:2009ia,
  Manschot:2010xp}. 

The change $\Delta\Omega(\Gamma_1+\Gamma_2;J_\CC\to J_{\CC'})$, for $\Gamma_1$ and $\Gamma_2$ primitive,
is \cite{Yoshioka:1996,Denef:2007vg}:
\begin{eqnarray}
\label{eq:deltaprimitive}
\Delta\Omega(\Gamma_1+\Gamma_2; J_\CC\to
J_{\CC'})&=&\half \left( \sgn(\mathcal{I}(\Gamma_1,\Gamma_2;J_{\CC'}))- \sgn(\mathcal{I}(\Gamma_1,\Gamma_2;J_\CC)) \right)\\
&&\times\, (-1)^{\left<\Gamma_1,\Gamma_2\right>}\left<\Gamma_1,\Gamma_2\right>
\Omega(\Gamma_1;J_{W_\CC})\,\Omega(\Gamma_2;J_{W_\CC}),\non
\end{eqnarray}
with 
$$
\sgn(x)=\left\{ \begin{array}{cl} 1, & x>0, \\ 0, & x=0, \\ -1, & x<0.\end{array} \right.
$$
The subscript ${W_\CC}$ in $J_{W_\CC}$ refers to a point in $\CC$
which is sufficiently close to the wall $W$, such that no wall is
crossed for the constituents between the wall and $J_{W_\CC}$. The wall
is independent of $c_2$, and therefore a sum over $c_2$ appears in the
next section.   

For the computation of the invariants for rank 3, one also needs the
semi-primitive wall-crossing formula \cite{Denef:2007vg}:  
\begin{eqnarray}
\label{eq:semiprim}
\Delta\Omega(2\Gamma_1+\Gamma_2;J_\CC\to
J_{\CC'})&=&\half \left( \sgn(\mathcal{I}(\Gamma_1,\Gamma_2;J_{\CC'}))- \sgn(\mathcal{I}(\Gamma_1,\Gamma_2;J_\CC)) \right)\non\\
&&\times \left\{-2\left<\Gamma_1, \Gamma_2 \right>
\Omega(2\Gamma_1;J_{W_\CC})\,\Omega(\Gamma_2;J_{W_\CC}) \right.\non\\
&&\quad +(-1)^{\left<\Gamma_1,\Gamma_2\right>}\left<\Gamma_1,\Gamma_2\right>\Omega(\Gamma_1;J_{W_\CC})\,\Omega(\Gamma_1+\Gamma_2; J_{W_\CC}) \\   
&&\quad \left. +\half \left<\Gamma_1, \Gamma_2\right>\Omega(\Gamma_2;J_{W_\CC})\,\Omega(\Gamma_1;J_{W_\CC})\,(\left<\Gamma_1, \Gamma_2\right>\Omega(\Gamma_1; J_{W_\CC})-1)\right\}.\non
\end{eqnarray}

Define the generating function for $\bar \Omega(\Gamma;J)$:
\be
\label{eq:genfunction}
h_{r,c_1}(\tau;S,J):=\sum_{c_2} \bar \Omega(\Gamma;J)\,q^{r\Delta-\frac{r\chi(S)}{24}}.
\ee
Twisting a sheaf by a line bundle gives an isomorphism of moduli
spaces, which implies $h_{r,c_1+\mathbf{k}}=h_{r,c_1}$ if
$\mathbf{k}\in H_2(S,r\mathbb{Z})$. It is therefore sufficient to
determine $h_{r,c_1}$ only for $c_1 \mod r$. Explicit computation of $h_{r,c_1}(\tau;S,J)$ is typically complicated. A generic
result exists just for $r=1$ \cite{Gottsche:1990}:
\be
\label{eq:genr=1}
h_{1,c_1}(\tau;S)=\frac{1}{\eta(\tau)^{\chi(S)}},
\ee 
with $\eta(\tau)$ defined in Eq. (\ref{eq:etatheta}). The dependence 
on $J$ could be omitted here, since the moduli space of rank 1 sheaves
does not depend on a choice of ample class. 

The next proposition gives
the universal relation between generating functions for $S$ and its
blow-up $\tilde S$. This appeared first for $r=2$ in
\cite{Yoshioka:1994, Vafa:1994tf} and for general $r$ in
\cite{Yoshioka:1996}. Proofs are given in \cite{Li:1999} for
$r=2$, and \cite{Gottsche:1998} for general $r$. 
\begin{proposition}
\label{prop:blowup}
Let $S$ be a smooth projective surface and $\phi: \tilde S \to S$ the
blow-up at a non-singular point, with $C_1$ the exceptional
divisor of $\phi$. Let $J\in C(S)$, $r$, and $c_1$ such that
$\gcd(r,c_1\cdot J)=1$. The generating functions $h_{r,c_1}(\tau;S,J)$ and
$h_{r,c_1}(\tau;\tilde S,J)$ are then related by the ``blow-up
formula'':
\be
h_{r,\phi^* c_1-kC_1}(\tau;\tilde S, \phi^* J)=B_{r,k}(\tau)\, h_{r,c_1}(\tau;S,J),\non
\ee
with
\be
B_{r,k}(\tau)=\frac{(-1)^{(r-1)k}}{\eta(\tau)^r}\sum_{\sum_{i=1}^ra_i=0 \atop a_i \in \mathbb{Z}+\frac{k}{r}} q^{-\sum_{i<j}a_ia_j}.\non 
\ee
\end{proposition}
The factor $(-1)^{(r-1)k}$ is a consequence of the relative sign between the
BPS-invariant and the Euler number. The two relevant cases for this article are $r=2,3$:
\be
B_{2,k}(\tau)=(-1)^k\frac{\sum_{n\in \mathbb{Z}+k/2} q^{n^2}}{\eta(\tau)^2},\qquad B_{3,k}(\tau)=\frac{\sum_{m,n \in \mathbb{Z}+k/3} q^{m^2+n^2+mn}}{\eta(\tau)^3}.
\ee

\subsection{Flow trees}
\label{subsec:flowtrees}
This subsection gives a brief introduction to flow trees, since the
computations in the next sections are inspired by it. More information
can be found in Refs. \cite{Denef:2000nb, Denef:2007vg}. See Ref.
\cite{Manschot:2010xp} for a discussion which is more adapted to the present context.

Flow trees appear in the analysis of D-brane bound
states. D-branes are equivalent to coherent sheaves in the ``infinite
volume limit''. A flow tree is an embedding of a rooted tree $T$ in $C(S)$ (or
more generally, the moduli space), which satisfies a number of
``stability'' conditions. The tree can be parametrized by a nested list,
e.g. $((\Gamma_1,\Gamma_2),\Gamma_3)$, and represents a decomposition  
of the total charge $\Gamma=\sum_{i}\Gamma_i$. The change of $J$ along the edges of the tree,
is determined by the supergravity equations of motion. The endpoints of the
flow tree represent ``elementary'' constituents which do not decay in
$C(S)$, for example rank 1 sheaves. Generically, only in a special
chamber in $C(S)$, the chamber with the attractor point, the total moduli space corresponds to the moduli space
of these elementary constituents. 
 
The existence of a tree as a flow tree is determined by ``stability conditions'' at its vertices. The class  
$J$ lies at a wall for the two merging trees if it is a trivalent vertex. For example, the stability of the subtree  $(\Gamma_1,\Gamma_2)$ in
$((\Gamma_1,\Gamma_2),\Gamma_3)$, is determined at a (specific) point of the wall for $\Gamma_1+\Gamma_2$ and $\Gamma_3$. If
all conditions are satisfied the tree does exist as a flow tree. The attractor flow conjecture states that the ``BPS Hilbert
space'' is partitioned by flow trees \cite{Andriyash:2010qv,
  Denef:2000nb}. This implies that the BPS-index can be computed in
principle by enumerating flow trees, once the BPS-indices of the
endpoints are known \cite{Denef:2007vg}.

One of the advantages of flow trees is that they give an algorithmic
procedure to test for the stability of a composite object at a given
point in moduli space. A simplifying feature is that they do not distinguish between subobjects and quotients, in 
contrast to the stratification of the set of sheaves using
(Harder-Narasimhan) filtrations. 

Small changes are necessary to utilize flow trees in the present
context, since the manifold is a surface instead of a 3-dimensional
Calabi-Yau threefold. One difference is the choice of the boundary of $C(\mathbb{\tilde P}^2)$ as reference
point in the moduli space, instead of the attractor point. For
$\mathbb{\tilde P}^2$ one does not need to
solve for $J$ along the edges, not even for the flow trees with
3 centers, since a wall in $C(\mathbb{\tilde P}^2)$ is only a single
point (projectively). With these observations, it is not difficult to
realize that the generating functions for the (refined) BPS-invariants
in Sections \ref{sec:euler} and  \ref{sec:betti} can be obtained
either using wall-crossing or enumeration of flow trees.  

\section{Euler numbers}
\label{sec:euler}

This section computes the generating function of Euler numbers of the
moduli spaces of semi-stable sheaves of rank 2 and 3 on $\mathbb{\tilde P}^2$ and
$\bP^2$. First, some rudiments of ruled surfaces are reviewed.  

\subsection{Some properties of ruled surfaces}
A ruled surface is a surface $\Sigma$ together with a surjective
morphism $\pi: \Sigma\to C$ to a curve $C$, such that the fibre $\Sigma_y$
is isomorphic to $\mathbb{P}^1$ for every point $y\in C$. Let $f$ be the
fibre of $\pi$, then
$H_2(\Sigma,\mathbb{Z})=\mathbb{Z}C\oplus\mathbb{Z}f$, with
intersection numbers $C^2=-e$, $f^2=0$ and $C\cdot f=1$. The canonical class is
$K_\Sigma=-2C+(2g-2-e)f$. The holomorphic Euler characteristic
$\chi(\CO_\Sigma)$ is for a ruled surface $1-g$. An ample
divisor is parametrized as $J_{m,n}=m(C+ef)+nf$ with $m,n\geq 1$.  

The blow-up $\phi: \mathbb{\tilde P}^2\to \mathbb{P}^2$ of the
projective plane $\mathbb{P}^2$ at a point is equal to the rationally
ruled surface with $(g,e)=(0,1)$. The exceptional divisor of $\phi$ is $C$, and
the hyperplane $H$ of $\bP^2$ equals $\phi^*(C+f)$. The remainder of
this article restricts to the case $(g,e)=(0,1)$, although many
results are easily generalized to generic $(g,e)$.

\subsection{Rank 2}
Our aim is to compute the generating function
$h_{2,c_1}(\tau;\mathbb{\tilde P}^2,J)$ \footnote{\, Since almost all generating
  series in this section are for $\mathbb{\tilde P}^2$, it is omitted
  from the arguments of $h_{r,c_1}$ in the following.}
defined by Eq. (\ref{eq:genfunction}). To learn about the set of
stable sheaves on $\mathbb{\tilde P}^2$ for $J\in C(S)$, it is
useful to first consider the restriction of the sheaves on $\mathbb{\tilde
  P}^2$ to $f$. Namely the restriction $E_{|f}$ is stable if and only if $E$ is $\mu$-stable for
$J=J_{0,1}$ and in the adjacent chamber \cite{Huybrechts:1996}. However, since every bundle of rank $\geq 2$ on
$\mathbb{P}^1$ is a sum of line bundles \cite{Grothendieck:1957},
there are no stable bundles with $r\geq 2$ on $\bP^1$. Therefore $\Omega(\Gamma;J_{0,1})=0$ for
$\Gamma=(r(F),-C-\alpha f,\mathrm{ch}_2)$ with $r(F)\geq 2$ and
$\alpha=0,1$. The computation for $c_1(E)=-\alpha f$ is more complicated, and is dealt with in the end of this subsection. 

To determine $h_{2,-C-\alpha f}(\tau;J_{m,n})$, one can either change
the polarization from $J_{0,1}$ to $J_{m,n}$ (see Figure \ref{fig:mspace}) 
and keep track of $\Omega(\Gamma;J)$ across the walls, or enumerate the flow trees for $J_{m,n}$. 
\begin{figure}[h!]
\centering
\includegraphics[totalheight=7cm]{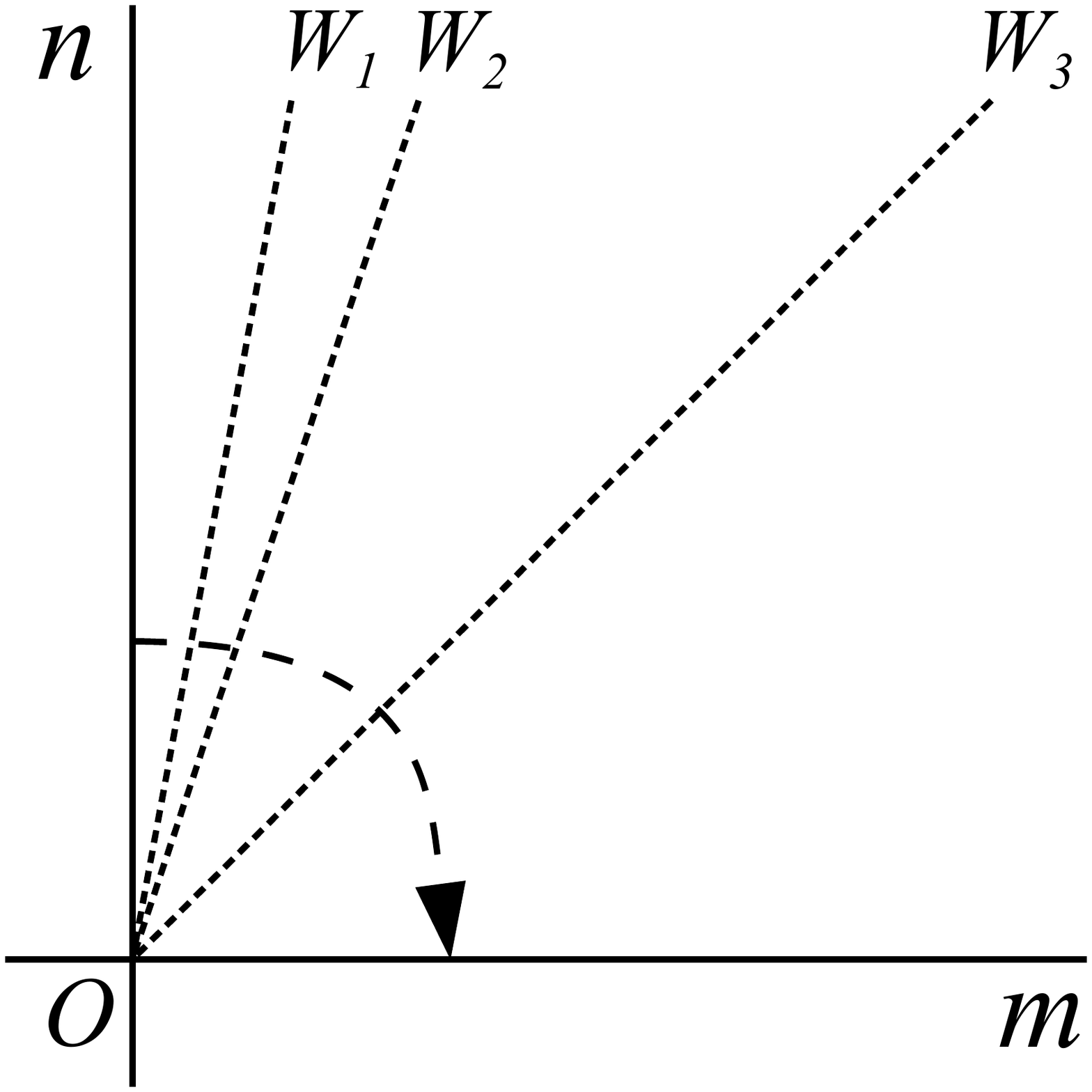}
\caption{The ample cone of $\mathbb{\tilde P}^2$, together with the
  three walls for $\Gamma=(2,-C-f,2)$, namely for $(a,b)=(1,0)$,
  $(2,0)$, $(3,0)$.} 
\label{fig:mspace}
\end{figure}
The only possible
filtrations are $0\subset F_1\subset F$, with $r_i=1$. Therefore the 
primitive wall-crossing formula (\ref{eq:deltaprimitive})
suffices, and the $\Omega(\Gamma_i)$ do not depend on the moduli. In
the following, $c_1(E_2)$ is parametrised by $bC-af$. As is customary for
flow trees, the constituents are treated symmetrically, such that $a,b$ run over $\mathbb{Z}$. Since $r_1=r_2$, this sum needs to be
multiplied by $\half$. The walls are then at $\frac{m}{n}=\frac{2b+1}{2a-\alpha}$, with $m,n\geq 0$. 
See Figure \ref{fig:mspace} for the walls for $\Delta(F)=\frac{9}{4}$, $r(F)=2$.

The various quantities appearing in
$h_{2,c_1}(\tau;J_{m,n})$ become in terms of $a$ and $b$:
\begin{eqnarray}
&\Delta(F)&=\half \Delta_1+\half \Delta_2+\frac{1}{8} (2b+1)^2+\frac{1}{4}(2b+1)(2a-\alpha),\non\\
&\left<\Gamma_1,\Gamma_2\right>&=-(2b+1)+2(2a-\alpha),\\
&\mathcal{I}(\Gamma_1,\Gamma_2;J_{m,n})&=(2b+1)n-(2a-\alpha)m.\non
\end{eqnarray}
It is now straightforward to construct the generating function using (\ref{eq:deltaprimitive}):
\begin{eqnarray}
\label{eq:genfunction2}
h_{2,-C-\alpha f}(\tau;J_{m,n})&=&-\frac{1}{2}\frac{1}{\eta(\tau)^8}\sum_{a,b\in\mathbb{Z}} \half
(\,\sgn((2b+1)n-(2a-\alpha)m)-\sgn(2b+1)\,)\\
&&\times\, (-(2b+1)+2(2a-\alpha))\,q^{\frac{1}{4} (2b+1)^2+\frac{1}{2}(2b+1)(2a-\alpha)}.\non
\end{eqnarray}
The $-$-sign in front is due to
$(-1)^{\left<\Gamma_1,\Gamma_2\right>}$, the $\half$ appears because
$r_1=r_2$, and $\eta(\tau)^{-8}$ arises from the sum over $\Delta_i$ and
(\ref{eq:genr=1}). Ref. \cite{Bringmann:2010sd} proved that for
$J_{m,n}=J_{1,0}$, the generating functions are 
\begin{eqnarray}
\label{eq:genseriesr21}
&&h_{2,-C-f}(\tau;J_{1,0})=3B_{2,0}(\tau) \mathfrak{h}_{1}(\tau)/\eta(\tau)^6,\\
&&h_{2,-C}(\tau;J_{1,0}) = 3B_{2,1}(\tau)
\mathfrak{h}_{0}(\tau)/\eta(\tau)^6, \non
\end{eqnarray}
where $\mathfrak{h}_{i}(\tau)$ are generating functions of the class
numbers (\ref{eq:genclassj}). The half-integer coefficients for
$c_1(F)=-C$ arise because $J_{1,0}$ is a wall for $F$. Application of Proposition \ref{prop:blowup} gives the known generating
functions for $\bP^2$ \cite{Vafa:1994tf}. This gives incidentally also
the correct result for $c_1(F)=0$, even though $\gcd(r(F),c_1(F)\cdot J)\neq 1 $.  

To compute $h_{3,-C-f}(\tau;J_{m,n})$, one needs explicit expressions
for $h_{2,-\alpha f}(\tau;J_{m,n})$, $\alpha=0,1$. Fortunately, it is not necessary to deal with the singularities
in the moduli space explicitly. One can either apply modular transformations or blow-down and blow-up again for $J_{m,n}=J_{1,0}$ and then
apply the wall-crossing formula. One
finds in both cases for $J_{1,0}$:
\begin{eqnarray}
\label{eq:genseriesr22}
&& h_{2,-f}(\tau;J_{1,0})= -3B_{2,1}(\tau)
\mathfrak{h}_{1}(\tau)/\eta(\tau)^6,\non\\
&& h_{2,0}(\tau;J_{1,0})=-3B_{2,0}(\tau)
\mathfrak{h}_{0}(\tau)/\eta(\tau)^6.\non
\end{eqnarray}
The Fourier coefficients $\bar \Omega(\Gamma; J_{1,0})$ of $h_{2,0}(\tau;J_{1,0})$ are not integers,
since $\Gamma$ might be divisible by 2. One finds for the generating
function of $\Omega(\Gamma;J_{1,0})$  using (\ref{eq:ratinv}):
\be
-3B_{2,0}(\tau) \mathfrak{h}_{0}(\tau)/\eta(\tau)^6-1/4\eta(2\tau)^4.\non
\ee
The wall-crossing formula provides now the generating functions for
generic $J\in C(\mathbb{\tilde P}^2)$:
\be
\label{eq:genseriesr2}
h_{2,\beta C-\alpha f}(\tau;J_{m,n})=h_{2,\beta C-\alpha
  f}(\tau;J_{1,0})+\Delta h_{2,\beta C-\alpha f}(\tau;J_{m,n}),
\ee
with
\begin{eqnarray}
\label{eq:Deltah}
&&\Delta h_{2,\beta C-\alpha f}(\tau;J_{m,n})=\non\\
&&\qquad (-)^{\beta}\frac{1}{2}\frac{1}{\eta(\tau)^8}\sum_{a,b\in\mathbb{Z}} \half
(\,\sgn(-(2a-\alpha )\,)-\sgn((2b-\beta )n-(2a-\alpha)m)\,)\\
&&\qquad \times\, (-(2b-\beta)+2(2a-\alpha))\,q^{\frac{1}{4} (2b-\beta)^2+\frac{1}{2}(2b-\beta)(2a-\alpha)}.\non
\end{eqnarray}

\subsection{Rank 3}
Using the results of the previous subsection, the Euler numbers of the
moduli space of stable sheaves with $\Gamma(F)=(3,-C-f,\mathrm{ch}_2)$ can  be computed. This
computation has to deal with two additional complications:
\begin{itemize}
\item[-] semi-primitive wall-crossing is possible for sheaves with $\Gamma(F)=2\Gamma_1+\Gamma_2$, 
\item[-] the BPS-invariants of a constituent with $r=2$ do themselves
  depend on the moduli, and need to be determined sufficiently close to the appropriate wall. 
\end{itemize}

Since no stable sheaves do exist for $c_1(F)=-C-f$, all sheaves
are composed of 2 constituents with rank $r_1=1$ and $r_2=2$, or 3 constituents
with rank $r_i=1$, $i=1,2,3$. Therefore the formulas of Ref. \cite{Manschot:2010xp} for
the enumeration of flow trees with 3 centers are applicable. There it was
explained that the semi-primitive wall-crossing formula for
$2\Gamma_1+\Gamma_2$ simplifies, if 1) the invariants are
evaluated at a point on the wall $J_W$ instead of $J_{W_\CC}$, and 2)  it is written in terms of the rational invariant $\bar \Omega(\Gamma,
J_W)$. With these substitutions, one finds that Eq. (\ref{eq:semiprim}) is equal to: 
\begin{eqnarray}
\Delta\Omega(2\Gamma_1+\Gamma_2;J_\CC \to J_{\CC'} )&=&\half \left(
  \sgn(\mathcal{I}(\Gamma_1,\Gamma_2;J_{\CC'}))-
  \sgn(\mathcal{I}(\Gamma_1,\Gamma_2;J_\CC)) \right)\left<\Gamma_1,\Gamma_2\right>\non\\
&&\times \left( -2\, \bar \Omega(2\Gamma
_1;J_W)\,\Omega(\Gamma_2)+(-1)^{\left<\Gamma_1,\Gamma_2\right>}\,\Omega(\Gamma_1)\,\Omega(\Gamma_1+\Gamma_2;J_W)\right).
\end{eqnarray}
One observes that the extra terms due to semi-primitive wall-crossing are naturally
included into the terms for primitive wall-crossing. 

The Euler numbers can now be obtained by simply implementing the
formulas. Choose again $c_1(E_2)=bC-af$. Then, the walls are at
\be
\frac{m}{n}=\frac{3b+2}{3a-2}, \qquad m,n \geq 0.\non
\ee
For the generating function follows: 
\begin{eqnarray}
\label{eq:genf2}  
h_{3,-C-f}(\tau;J_{m,n})&=&\frac{1}{\eta(\tau)^{4}}\sum_{a,b\in
  \mathbb{Z}}\half (\, \sgn((3b+2)n-(3a-2)m)-\sgn(3b+2)
\,)\non \\
&&\times (-1)^b\,(-(3b+2)+2(3a-2))\,q^{\frac{1}{12}(3b+2)^2+\frac{1}{6}(3b+2)(3a-2)}\\
&& \times h_{2,bC-af}(\tau;J_{|3b+2|,|3a-2|}).\non
\end{eqnarray}
Expansion of the first coefficients gives for $J_{m,n}=J_{1,0}$: 
\be
h_{3,-C-f}(\tau;J_{1,0})=q^{-\frac{5}{6}}\left(3\,q^2
  + 69\,q^{3}+792\,q^4+6345\,q^5+\dots\right).
\ee
One finds with Proposition \ref{prop:blowup}:
\be
\label{eq:euler3p2}
h_{3,-H}(\tau;\mathbb{P}^2)=\frac{h_{3,-C-f}(\tau;J_{1,0})}{B_{3,0}(\tau)},
\ee
which is also equal to $h_{3,H}(\tau;\mathbb{P}^2)$.  
This result agrees with the coefficients given by Corollary 4.10 of
Ref. \cite{weist:2009},\footnote{\, Note that the result of
  Ref. \cite{weist:2009} differs from (\ref{eq:euler3p2}) by
  $\eta(\tau)^{-9}$, since that article considers vector bundles
  instead of sheaves.} and Corollary 4.9 of Ref. \cite{Kool:2009}.

\section{Betti numbers}
\label{sec:betti}
This section computes the Betti numbers of the moduli spaces of
stable sheaves with $\Gamma(F)=(3,-C-f,\mathrm{ch}_2)$ using
wall-crossing for refined (or motivic) invariants
$\Omega(\Gamma,w;J)$. To define these invariants, let
$p(X,s)=\sum_{i=0}^{2\dim_\mathbb{C}(X)}b_is^i$, with $b_i$
the Betti numbers $b_i=\dim H^2(X,\mathbb{Z})$, be the Poincar\'e
polynomial of a compact complex manifold $X$. Then I define the refined
invariant in terms of the Betti numbers by:
\be
\Omega(\Gamma,w;J):=\frac{w^{-\dim_\mathbb{C}\CM_J(\Gamma)}}{w-w^{-1}}\, p(\CM_J(\Gamma),w).\non
\ee
The primitive wall-crossing formula reads for $\Omega(\Gamma,w;J)$ \cite{Yoshioka:1996}:
\begin{eqnarray}
\Delta \Omega(\Gamma,w;J_\CC\to J_{\CC'})&=&-\half \left( \sgn(\mathcal{I}(\Gamma_1,\Gamma_2;J_{\CC'}))- \sgn(\mathcal{I}(\Gamma_1,\Gamma_2;J_\CC)) \right)\non\\
&&\times\left( w^{\left<\Gamma_1,\Gamma_2\right>}- w^{-\left<\Gamma_1,\Gamma_2\right>}\right)\Omega(\Gamma_1,w;J)\,\Omega(\Gamma_2,w;J).\non
\end{eqnarray}

Using the semi-primitive wall-crossing formula for refined invariants
\cite{Dimofte:2009bv}, it becomes clear that the analogue of $\bar
\Omega(\Gamma;J)$ for refined invariants is:
\be
\label{eq:refw}
\bar \Omega(\Gamma,w;J)=\sum_{m|\Gamma} (-1)^{m+1}\frac{\Omega(\Gamma/m,w^m;J)}{m}.\non
\ee
The generating function is naturally defined by:
\be
h_{r,c_1}(z,\tau;S,J)=\sum_{c_2} \bar \Omega(\Gamma,w;J)\, q^{r\Delta(F)-\frac{r\chi(S)}{24}}.
\ee
Note that the power of the denominator in (\ref{eq:refw}) is 1 whereas it was 2 in
Eq. (\ref{eq:ratinv}). This leads to an interesting product formula when an additional sum over the rank is
performed. The generalization of Proposition \ref{prop:blowup} gives \cite{Yoshioka:1996}:
\be
\label{eq:blowupw}
B_{2,k}(z,\tau)=\frac{\sum_{n\in \mathbb{Z}+k/2}
  q^{n^2}w^{2n}}{\eta(\tau)^2},\qquad B_{3,k}(z,\tau)=\frac{\sum_{m,n\in \mathbb{Z}+k/3}q^{m^2+n^2+mn}w^{4m+2n}}{\eta(\tau)^3}.
\ee
The generating function of refined invariants for $\mathbb{\tilde P}^2$ and $r=1$ is:
\be
h_{1,c_1}(z,\tau)=\frac{i}{\theta_1(2z,\tau)\,\eta(\tau)}.\non
\ee

Now the computation is completely analogous to Section
\ref{sec:euler}. The generalization of Eq. (\ref{eq:genfunction2}) is: 
\begin{eqnarray}
\label{eq:genf2w}
h_{2,-C-\alpha f}(z,\tau;J_{m,n})&=&\half \frac{1}{\theta_1(2z,\tau)^2\,\eta(\tau)^2}\sum_{a,b\in
  \mathbb{Z}} \half (\,\sgn((2b+1)n-(2a-\alpha)m)-\sgn(2b+1)\,)\non\\
&&\times \left(w^{- (2b+1)+2(2a-\alpha)}-w^{
    (2b+1)-2(2a-\alpha)} \right)\,q^{\frac{1}{4}
  (2b+1)^2+\frac{1}{2}(2b+1)(2a-\alpha)}. 
\end{eqnarray}
This gives for $J_{m,n}=J_{1,0}$:
\be
\begin{array}{ll}
h_{2,-C-f}(z,\tau;J_{1,0}) & = -B_{2,0}(z,\tau)\, g_{1}(z,\tau)/\theta_1(2z,\tau)^2, \non\\
h_{2,-C}(z,\tau;J_{1,0}) & = -B_{2,1}(z,\tau)\,g_0(z,\tau)/\theta_1(2z,\tau)^2, \\
h_{2,-f}(z,\tau;J_{1,0}) & = -B_{2,1}(z,\tau) \, g_{1}(z,\tau) / \theta_1(2z,\tau)^2, \non \\
h_{2,0}(z,\tau;J_{1,0}) & = -B_{2,0}(z,\tau) \,g_0(z,\tau)/\theta_1(2z,\tau)^2.\non
\end{array}
\ee
The invariants for generic $J_{m,n}$ are obtained using the 
generalization of $\Delta h_{2,\beta C-\alpha a}(\tau; J_{m,n})$
(\ref{eq:Deltah}). This gives for rank 3: 
\begin{eqnarray}
\label{eq:h3Cf}
  h_{3,-C-f}(z,\tau;J_{m,n})&=&-\frac{i}{\theta_1(2z,\tau)\,\eta(\tau)}\sum_{a,b\in \mathbb{Z}}\half (\, \sgn((3b+2)n-(3a-2)m)-\sgn(3b+2)
\,)\non \\
&&\times \left(w^{- (3b+2)+2(3a-2)}-w^{ (3b+2)-2(3a-2)} \right)\,q^{\frac{1}{12}(3b+2)^2+\frac{1}{6}(3b+2)(3a-2)}\\
&& \times\, h_{2,bC-af}(z,\tau;J_{|3b+2|,|3a-2|}).\non
\end{eqnarray}

With Eq. (\ref{eq:blowupw}) for rank 3, the final result for $\bP^2$ is:
\be
\label{eq:r3w}
h_{3,-H}(z,\tau;\bP^2)=\frac{h_{3,-C-f}(z,\tau;J_{1,0})}{B_{3,0}(z,\tau)}. 
\ee
The Betti numbers for $2\leq c_2\leq 6$ are presented in Table
\ref{tab:betti}. The first three lines agree with the three
Poincar\'e polynomials presented by Yoshioka \cite{Yoshioka:1996}.

\begin{table}[h!]
\begin{tabular}{lrrrrrrrrrrrrrrr}
$c_2$ & $b_0$ & $b_2$ & $b_4$ & $b_6$ & $b_8$ & $b_{10}$ & $b_{12}$ & $b_{14}$
& $b_{16}$ & $b_{18}$ & $b_{20}$ & $b_{22}$ & $b_{24}$ & $b_{26}$ & $\chi$ \\
\hline
2 & 1 & 1 & & & & & & & & & & & & & 3 \\
3 & 1 & 2 & 5 & 8 & 10 & & & & & & & & & & 42  \\
4 & 1 & 2 & 6 & 12 & 24  & 38 & 54 & 59 & &&&&& & 333\\
5 & 1 & 2 & 6 & 13 & 28 & 52 & 94 &  149 & 217 & 273 & 298 & && & 1968\\
6 & 1 & 2 & 6 & 13 & 29 & 56 & 108 & 189 & 322 & 505 & 744 & 992 &
1200 & 1275 & 9609 \\
&
\end{tabular}
\caption{The Betti numbers $b_n$ (with $n\leq
  \dim_\mathbb{C} \mathcal{M}$) and the Euler number $\chi$ of the moduli spaces of stable sheaves
  on $\bP^2$ with $r=3$, $c_1=-H$, and $2\leq c_2\leq 6$.}  
\label{tab:betti}
\end{table}
  
Note that Eq. (\ref{eq:r3w}) is rather compact and
expressed in terms of modular functions. Electric-magnetic duality suggests that
$h_{3,-H}(z,\tau;\bP^2)$ exhibits modular transformation properties. Indeed, one
observes a convergent sum over a subset of an indefinite lattice of signature
$(2,2)$, when one substitutes the explicit expression for $h_{2,bC-af}(z,\tau;J_{|3b+2|,|3a-2|})$ in Eq. (\ref{eq:h3Cf}). Similar sums over lattices of signature $(n,1)$
appeared earlier in the literature for rank 2 sheaves \cite{Gottsche:1996,
  Gottsche:1998}, which can also be seen from Eq. (\ref{eq:genf2w}). A detailed discussion
of the modular properties of $h_{3,-H}(z,\tau;\bP^2)$ and the
computation of $h_{3,0}(z,\tau;\mathbb{P}^2)$ will appear in a future
article \cite{Bringmann:2011}.

\appendix

\section{Modular functions}
\label{app:modfunctions}

This appendix lists various modular functions, which appear in the
generating functions in the main text. Define $q:=e^{2\pi i \tau}$,
$w:= e^{2\pi i z}$, with $\tau \in \mathbb{H}$ and $z\in \mathbb{C}$. 
The Dedekind eta and Jacobi theta functions are defined by: 
\begin{eqnarray}
\label{eq:etatheta}
&&\eta(\tau)\quad \,\,:=q^{\frac{1}{24}}\prod_{n=1}^\infty (1-q^n),\non\\
&&\theta_1(z,\tau):=i\sum_{r\in \mathbb{Z}+\frac{1}{2}} (-1)^{r-\frac{1}{2}} q^{\frac{r^2}{2}}w^{r},\\
&&\theta_2(z,\tau):=\sum_{r\in \mathbb{Z}+\frac{1}{2}} q^{r^2/2}w^{r},\non\\
&&\theta_3(z,\tau):=\sum_{n\in\mathbb{Z}}q^{n^2/2}w^n\non.
\end{eqnarray}

Let $H(n)$ be the Hurwitz class number, i.e., the number of equivalence classes of quadratic forms of 
discriminant $-n$, where each class $C$ is counted with multiplicity 1/Aut($C$).
Define the generating functions of the class numbers \cite{zagier:1975}:
\be
\label{eq:genclassj}
\mathfrak{h}_j(\tau):=\sum_{n=0}^\infty
H(4n+3j)\,q^{n+\frac{3j}{4}},\qquad j\in \{0,1 \}.
\ee

Following Ref. \cite{Bringmann:2010sd}, define:
\begin{eqnarray}
g_{0}(z,\tau)&:=&\frac{1}{2}+\frac{q^{-\frac{3}{4}}w^{5}}{\theta_2(2z,2\tau)}\sum_{n\in
\mathbb{Z}}\frac{q^{n^2+n}w^{-2n}}{1-q^{2n-1}w^4},\\
g_1(z,\tau)&:=&\frac{q^{-\frac{1}{4}}w^3}{\theta_3(2z,2\tau)}\sum_{n\in
\mathbb{Z}}\frac{q^{n^2}w^{-2n}}{1-q^{2n-1}w^4}.\non
\end{eqnarray}

\providecommand{\href}[2]{#2}\begingroup\raggedright\endgroup

\end{document}